\documentclass[11pt]{article}
\usepackage{setspace}
\usepackage{amsmath}
\usepackage{amsfonts}
\usepackage{geometry}
\usepackage{hyperref}
\usepackage{natbib}
\usepackage{fancyhdr}
\usepackage{titlesec}
\usepackage{lipsum} 

\usepackage{amssymb}
\usepackage{amsfonts}
\usepackage[overload]{textcase}
\usepackage[T1]{fontenc}
\usepackage[utf8]{inputenc}
\usepackage{longtable}
\usepackage[pdftex]{graphicx}

\usepackage{caption}
\usepackage{subcaption}

\usepackage{color}
\definecolor{mygreen}{rgb}{0,0.6,0}
\definecolor{myblue}{rgb}{0.3,0.2,0.8}
\definecolor{myred}{rgb}{0.8,0.1,0.1}

\usepackage[american]{babel}
\usepackage{csquotes}
\usepackage{mathrsfs}
\usepackage{float}
\usepackage{amsthm}
\newtheorem{theorem}{Theorem}
\usepackage{enumerate}
\usepackage{enumitem}

\newtheorem{corollary}{Corollary}
\newtheorem{lemma}{Lemma}

\def\hjse{h^{\rm JSE}}
\def\sjse{\Sigma^{\rm JSE}}
\def\smjse{\Sigma^{\rm MJSE}}
\def\sms{\Sigma^{\rm MS}}

\def\djse{\Delta^{\rm JSE}}
\def\bfo{\mathbf{1}}
\def\bfz{\mathbf{0}}
\def\sjsm{\Sigma^{\rm JSM}}
\def\hjsm{H_{\rm JSM}}

\titleformat{\section}[block]{\normalfont\Large\bfseries}{\thesection}{1em}{}
\titleformat{\subsection}[runin]{\normalfont\bfseries}{\thesubsection}{1em}{}
\titleformat{\subsubsection}[runin]{\normalfont}{\thesubsubsection}{1em}{}

\title{\textbf{Understanding the Long-Only Minimum Variance Portfolio} }
\author{Nicholas Gunther \\
Consortium for Data Analytics in Risk (CDAR), UC Berkeley \\
{nlgunther@gmail.com}\\
Alec Kercheval \\ Dept. of Mathematics, Florida State University and CDAR, UC Berkeley \\ 
{akercheval@fsu.edu} \\
Ololade Sowunmi \\
Dept. of Mathematics, Florida State University \\
{osowunmi@fsu.edu} 
}

\date{First version: 10 Dec 2024; This version: 9 July 2025}

\begin{document}

\maketitle
\thispagestyle{empty} 

\begin{abstract}
For a covariance matrix coming from a factor model of returns, we investigate the relationship between the long-only global minimum variance portfolio and the asset exposures to the factors.  In the case of a 1-factor model, we provide a rigorous and explicit description of the long-only solution in terms of the parameters of the covariance matrix. For $q>1$ factors, we provide a description of the long-only portfolio in geometric terms.  The results are illustrated with empirical daily returns of US stocks.


\end{abstract}


\section{Long-only minimum variance}\label{sec:1}

\subsection{Introduction}

There is a long-standing interest in equity portfolios optimized to have the lowest possible variance. The optimal such portfolio depends on the covariances between pairs of assets, and on the particular constraints of interest.

If there are $p$ assets available for investment, we denote by $w= (w_1,\dots,w_p)^\top \in \mathbf{R}^p$ the $p$-dimensional vector of asset weights defining the portfolio.  The global minimum variance portfolio $w^{LS}$ denotes the long-short portfolio solving the simplest problem
 \begin{equation}\label{eq:LSprob}
 	      \begin{split}		
 		 &\min_{w \in \mathbf{R}^p}  w^\top \Sigma w \\
                &w^\top\bfo_p= 1,
 	      \end{split}
\end{equation}
where $\Sigma$ denotes the positive definite covariance matrix of asset returns; $\bfo_p$ denotes the vector of dimension $p$ whose every entry is 1; 
and $w^\top \bfo_p =1$ is the full investment condition setting the sum of the weights equal to 1. 
For the long-short portfolio, some of the weights may be negative.

 Our focus is the long-only minimum variance (LOMV) problem 
 \begin{equation}\label{eq: prob 1}
 	      \begin{split}		
 		 &\min_{w \in \mathbf{R}^p}  w^\top \Sigma w \\
                &w^\top\bfo_p= 1\\
             &w_i \geq 0 \text{ } \text{for all  $i=1,2,...,p$.}
 	      \end{split}
\end{equation}
The long-only constraints $w_i \geq 0$ are often required for real investment portfolios due to the complications and costs of short positions.

The solution ${w =  w^L}$ of \eqref{eq: prob 1} represents the long-only fully invested global minimum risk portfolio, and can be contrasted with the solution $w^{LS}$ of the long-short problem. 
The portfolio $w^{LS}$ solving problem \eqref{eq:LSprob}  is given by the
simple formula
\begin{equation}
    w^{LS} = \frac{\Sigma^{-1}\bfo_p}{\bfo_p^\top\Sigma^{-1}\bfo_p}.
\end{equation}
The long-only problem \eqref{eq: prob 1} is less straightforward.

As we show in this article, the main difficulty in problem \eqref{eq: prob 1} is determining which are the active (positive weight) assets in the optimal portfolio, or, equivalently, which are the assets for which the long-only constraints are binding. 
Denote by $K = \{i \leq p: w^L_i >0\}$ the set of active assets in the long-only optimal  portfolio, and let $k \leq p$ denote the number of elements of $K$ . Once we have determined $K$, Theorem \ref{thm:wL} solves the problem: if we denote by $\Sigma^K$ the $k \times k$ matrix obtained from $\Sigma$ by deleting all the rows and columns not in $K$, then the $k$ positive weights of $w^L$ are given by the corresponding entries of the $k$-dimensional vector
\begin{equation}
    w^{K} = \frac{(\Sigma^K)^{-1}\bfo_k}{\bfo_k^\top(\Sigma^K)^{-1}\bfo_k}.
\end{equation}

In short, {\em the active long-only minimum risk portfolio holdings are those of the long-short minimum risk portfolio corresponding to a reduced set of available assets.}  This is intuitively reasonable\footnote{Geometrically, if, for example, the active LOMV assets are the first $k$ in the list, then the corresponding $k$-vector minimizes $k$-dimensional variance in the interior of the positive $k$-orthant.} because the long-only constraints are not binding on the positive holdings of $w^L$.

This leaves the problem of determining $K$, the set of active assets of $w^L$, which is normally much smaller than the set of positive-weight assets in the long-short portfolio $w^{LS}$ (e.g. Figure \eqref{fig:w-v-beta}). In this article we analyze that problem when the covariance matrix comes from a factor model.
In the case of a single-factor model, we provide an essentially explicit description of $K$ below in Theorem \ref{thm:1}. There, one factor explains covariance between assets, and the covariance matrix takes the form 
\begin{equation} 
\Sigma = \sigma^2 \beta \beta^\top + \Delta,
\end{equation}
where $\sigma^2 >0$ is the factor return variance,  $\beta$ is a $p$-vector of exposures to the  factor, and $\Delta$ is a diagonal matrix of specific variances.

In the case of a multiple factor model with $q >1$ factors, the $p \times p$ covariance matrix takes the form
\begin{equation}
    \Sigma = B \Omega B^\top + \Delta,
\end{equation}
where $B$ is a $p \times q$ matrix whose columns are the asset exposures to each of $q$ factors, and $\Omega$ is a $q \times q$ invertible matrix of factor variances, diagonal in the case when the columns of $B$ are principal components. When $q >1$, we know of no direct way to compute the set $K$ of active assets that is similar to Theorem \ref{thm:1}. However, Theorem \ref{thm:multi} gives us a necessary condition satisfied by $K$, as follows. Let $B_i \in \mathbf{R}^q$ denote the $i$th row of $B$. Then there is a $(q-1)$-dimensional hyperplane $H$ in $\mathbf{R}^q$ such that  the elements $i$ of $K$ are all those such that $B_i$ lies on the same side of $H$ as the origin.

The single factor case was studied in \cite{cst2011}, where they assume the single factor is the market return, and provide an implicit solution to the long-only problem that is equivalent to the one here when our vector beta of exposures to the single factor is taken to be the market beta. Their article inspired our work, and we discuss the relationship between their results and ours further below.

\subsection{Main Results}

We state the results outlined above in more detail in this section.  The proofs of the following theorems appear in Section \ref{sec:proof}.
\medskip

{\bf Assumption for Theorem \ref{thm:wL}.} Suppose that the covariance matrix $\Sigma$ of returns for a universe of $p$ assets is an arbitrary $p \times p$ symmetric positive definite matrix. 

{
\begin{theorem} \label{thm:wL}
 Denote by $w^L$ the solution of problem \eqref{eq: prob 1} and let $K$ denote the set of active assets in $w^L$:
    \begin{equation}
        K = \{i \leq p : w^L_i > 0\},
    \end{equation}
and $k = |K| \leq p$, the number of active assets.
 Let $\Sigma^{K,0}$ be the modified matrix obtained from $\Sigma$ by setting to zero the rows and columns not belonging to $K$.

Then  
\begin{equation}
    w^L  = \frac{(\Sigma^{K,0})^+\bfo_p}{\bfo_p^\top(\Sigma^{K,0})^+\bfo_p}
\end{equation}
where $^+$ denotes the Moore-Penrose inverse.\footnote{For a symmetric matrix $S$ with singular value decomposition $S = UDU^\top$ for orthogonal $U$ and diagonal $D$, the Moore-Penrose inverse is defined by $S^+ = UD^+U^\top$, where the diagonal matrix $D^+$ is obtained by replacing each nonzero element by its inverse.}
\end{theorem}

Equivalently, if we denote by  $w^K$ the $k$-vector obtained by deleting all the zero entries of the $p$-vector $w^L$, and $\Sigma^K$ denotes the $k \times k$ matrix obtained from $\Sigma$ by deleting the rows and columns not belonging to $K$, then
\begin{equation}
    w^K  = \frac{(\Sigma^K)^{-1}\bfo_k}{\bfo_k^\top(\Sigma^K)^{-1}\bfo_k}.
\end{equation}
This is the unique solution of the $k$-dimensional long-short problem \eqref{eq:LSprob} with
$\Sigma$ replaced by $\Sigma^K$, and $w^L$ can be recovered from $w^K$ by adding back zero entries for the deleted assets.

}


    

Theorem \ref{thm:wL} allows us to determine the long-only minimum risk portfolio as soon as we know the set $K$ of active assets. It remains only to determine $K$, and for the single-index case $q=1$ this can be accomplished by an explicit method described next in Theorem \ref{thm:1}.
\medskip

{\bf Assumptions for Theorem \ref{thm:1}.} Assume now a single-factor returns model which gives the covariance matrix $\Sigma$ the form
\begin{equation} \label{eq:singlefactor}
\Sigma = \sigma^2 \beta \beta^\top + \Delta,
\end{equation}
where $\sigma^2 >0$ is the factor return variance,  $\beta$ is a $p$-vector of factor exposures, and $\Delta= diag(\delta_1^2,\delta_2^2,...,\delta_p^2)$ is a diagonal matrix of non-zero idiosyncratic asset return variances $\delta_i^2 > 0$. 

We permit the entries $\beta_i$ of $\beta$ to be positive, negative, or zero, but assume the generic condition
\begin{equation} 
\sum_{i=1}^p \frac{\beta_i}{\delta_i^2} \neq 0.
\end{equation}

Notice that the vector $\beta$ may be replaced by $-\beta$ without changing $\Sigma$, so without loss of generality we choose the sign so that 
\begin{equation} \label{eq:positive}
 \sum_{i=1}^p \frac{\beta_i}{\delta_i^2} > 0, 
\end{equation}
as would be the case if all the betas were positive.

In addition, by re-ordering the assets if necessary, for convenience we further assume without loss of generality that the betas are arranged in increasing order:
\begin{equation}
    \beta_1 \leq \beta_2 \leq \cdots \leq \beta_p.
\end{equation}

With these assumptions, our main result is
\begin{theorem} \label{thm:1}
     (\textbf{Explicit Solution to the long-only constrained problem under a 1-factor model})
Assume the betas $\beta_i$ are arranged in increasing order.     
\begin{enumerate}
    \item Let $R_1 = \frac{1}{\sigma^2}$, and for  $2 \leq i \leq p$, let 
\[
R_i = {\frac{1}{\sigma^2} + \sum_{j=1}^{i-1}\frac{\beta_j}{\delta_j^2}(\beta_j - \beta_i)}.
\]
Then there exists $s \leq p$ such that the initially positive sequence $\{R_i\}$ is monotonically increasing until $i = s$, then monotonically  decreasing for $i > s$. That is,
\begin{equation}
    0 < R_1 \leq R_2 \leq \cdots \leq R_s \geq R_{s+1} \geq \cdots \geq R_p.
\end{equation}
In particular, the  sequence $\{R_i\}$ crosses zero at most once.    

\item Let $w^L$ denote the solution of problem \eqref{eq: prob 1}, $K = \{i \leq p : w^L_i >0\}$, and $k = |K|$.  

Then
\begin{equation} \label{eq:kdef}
   k = \max\big\{i \le p:R_i >0\big\} \text{ and } K = \{1,2,\dots, k\}.
\end{equation}

 Applying Theorem \ref{thm:wL}, we may conclude the following. Let $\Sigma^{K,0}$ be the block-diagonal matrix obtained from $\Sigma$ by replacing all but the first $k$ rows and columns with zeros.  Then the solution $w^L$ is given by
\begin{equation}
    w^L  = \frac{(\Sigma^{K,0})^+\bfo_p}{\bfo_p^\top(\Sigma^{K,0})^+\bfo_p}
\end{equation}
where $+$ denotes the Moore-Penrose inverse.

\end{enumerate}
\end{theorem}
 
For an equivalent formulation in terms of ordinary matrix inverse, let $\Sigma^K$ be the $k \times k$ submatrix of $\Sigma$ consisting of the first $k$ rows and columns. Denote by $w^K$,
\begin{equation} \label{eq:closed_solution}
    w^K = \frac{ (\Sigma^{K})^{-1}\bfo_k}{\bfo_k^\top(\Sigma^{K})^{-1}\bfo_k},
\end{equation}
the long-short fully invested minimum variance solution for the first $k$ assets.

Then the solution $w^L = (w_1^L,\dots,w_p^L)^\top$ of  \eqref{eq: prob 1}   is given by
 \begin{equation}
            \begin{split}
                &w^L_i = w^K_i \text{  for $i=1,...,k$}\\
                &w^L_i=0 \text{ for  } i=k+1,\dots,p.\\
            \end{split}
\end{equation}

We note that $k = p$ if the long-short fully invested minimum variance portfolio happens to be long-only already. Otherwise, $k < p$, $R_p < 0$, and the sequence $\{R_i\}$ crosses zero exactly once.  In this situation, $k$ has the property 
\begin{equation}
 R_k > 0 \mbox{ and }    R_{k+1} = R_k + (\beta_k-\beta_{k+1})\sum_{j=1}^k \frac{\beta_j}{\delta_j^2} \leq 0.
\end{equation}
This means that the threshold index $k$ and the solution $w$ are not influenced by the values of $\delta_j$ for $j>k$, nor by the values of $\beta_j$ for $\beta_{j} > \beta_{k+1}$. 

The proof of Theorem \ref{thm:1}  also establishes, via Lemma \ref{lem:threshold} below, the following 
\begin{corollary} \label{cor:threshold}
    Under the assumptions above, if $w$ is the solution of problem \eqref{eq: prob 1}, then
\begin{equation} \label{eq:thresholdL}
   w_i > 0 \mbox{ if and only if } \beta_i < \frac{\frac{1}{\sigma^2}+\sum_{j=1}^k \frac{\beta^2_j}{\delta_j^2}}{\sum_{j=1}^k\frac{\beta_j}{\delta_j^2}}.
\end{equation}
\end{corollary}

The solution of problem \eqref{eq: prob 1} in semi-explicit form was previously described by R. Clarke, H. de Silva, and S. Thorley in \cite{cst2011}. They give the following condition, which is closely related to
Corollary \ref{cor:threshold}:
\begin{equation} 
    w_i > 0 \mbox{ if and only if } \beta_i < \tau,
\end{equation}
where $\tau$ is the solution of the equation
\begin{equation}
    \tau = \frac{\frac{1}{\sigma^2}+\sum_{\beta_i<\tau} \frac{\beta^2_j}{\delta_j^2}}{\sum_{\beta_i< \tau}\frac{\beta_j}{\delta_j^2}}.
\end{equation}
The solution described in  \cite{cst2011}  is equivalent to \eqref{eq:closed_solution} and \eqref{eq:thresholdL}. Both solutions require the assumption \eqref{eq:positive}, as a simple argument shows:
 In the absence of any hypotheses about the signs of the betas, replacing $\beta$ by $-\beta$ leaves the covariance matrix $\Sigma$, and hence the solution $w$, unchanged. But in that case the long positions of $w$ would correspond to betas above a threshold, not below it, contradicting the threshold condition.

The need for assumption \eqref{eq:positive} is easily overlooked because it likely always holds for betas actually observed in the market. In our one-factor world, if the betas are defined relative to a benchmark portfolio $w_B$ that belongs to our investable universe of $p$ assets, i.e. 
\begin{equation}
  w_B \in \mathbf{R}^p \mbox{ and }  \beta = \frac{\Sigma w_B}{\sigma_B^2},
\end{equation}
then a calculation similar to the ones in Section \ref{sec:proof} shows that \eqref{eq:positive} holds if and only if the benchmark $w_B$ is net long, $w_B^\top \bfo_p >0$.
This will be the case for any reasonable benchmark.

The closed form solution \eqref{eq:closed_solution} depends on first determining $k$ from \eqref{eq:kdef}. From the monotonicity properties of $\{R_i\}$, this may be quickly accomplished, for example, by the bisection method in $O(\log p)$ steps.
\medskip

{\bf The multifactor case.} When there are $q>1$ factors, there is no simple way to order the betas or immediately use them directly to identify the active assets in the long-only portfolio.  The next theorem tells us that we do know something about them: their corresponding betas are exactly those that lie in a particular half-space in $\mathbf{R}^q$.
\medskip

{\bf Assumptions for Theorem \ref{thm:multi}.} 
We assume asset returns follow a $q$-factor model, so that the covariance matrix of asset returns takes the form of the $p \times p$ matrix
\begin{equation} \label{eq:multi-cov}
    \Sigma = B\Omega B^\top + \Delta,
\end{equation}
where $B$ is a $p \times q$ matrix whose columns are the asset exposures to the $q$ factors, $\Omega$ is a $q \times q$ diagonal matrix of positive factor variances, and $\Delta$ is a $p \times p$ diagonal matrix of positive asset specific variances.

As before, for any subset $K$ of the indices
$\{1,2,\dots,p\}$, if $k = |K|$, we denote by $\Sigma^{K,0}$ the $p \times p$ matrix obtained from $\Sigma$ by setting all the rows and columns not in $K$ to zero, $\Sigma^K$ the $k \times k$ principal submatrix obtained by deleting the rows and columns not in $K$, and $B^K$ the $k \times q$ matrix obtained by deleting the rows of $B$ not in $K$.

\begin{theorem}[\bf Hyperplane separation for $q$-factor models]
\label{thm:multi}

Suppose that $w^L$ is the solution of problem \eqref{eq: prob 1} for the covariance matrix \eqref{eq:multi-cov}.  Let
\begin{equation}
    K = \{i \leq p: w^L_i >0 \}
\end{equation}
and $k = |K|$ the cardinality of $K$.

Define the $q$-dimensional column vector $h^K$ by
\begin{equation} \label{eq:hk}
    h^K = \Omega (B^K)^\top (\Sigma^K)^{-1} \bfo_k = \Omega B^\top (\Sigma^{K,0})^+ \bfo_p.
\end{equation}
Then 
\begin{equation}
    w^L_i > 0 \text{ if and only if }  B_i h^K < 1,
\end{equation}
where $B_i$ the $q$-dimensional $i$th row of $B$.
\end{theorem}

Geometrically, the condition of Theorem \ref{thm:multi} can be described  in terms of the hyperplane $H$ of $\mathbf{R}^q$ defined by
\begin{equation}
    H = \{x \in \mathbf{R}^q: x^\top h^K = 1 \}.
\end{equation}
It says that the assets included in the long-only optimal portfolio are those whose factor exposure vectors $B_i$ in $\mathbf{R}^q$ lie on the same side of $H$ as the origin.  See Figure \ref{fig:separation} for an illustration when $q=2$.

{
The next corollary gives alternative hyperplane formulation that is sometimes useful.
\begin{corollary} \label{cor:alt}
    With the same assumptions and notation of Theorem \ref{thm:multi}, let
\begin{eqnarray}\label{eq:sephyp}
    h^L &=& \left(\Omega^{-1} + (B^K)^\top (\Delta^K)^{-1} B^K\right)^{-1} (B^K)^\top (\Delta^K)^{-1} \bfo_k \\
    &=& \left(\Omega^{-1} + B^\top (\Delta^{K,0})^+ B \right)^{-1} B^\top (\Delta^{K,0})^+ \bfo_p.
\end{eqnarray}
Then $B^K h^L = B^K h^K$, and hence
    $w^L_i > 0$  if and only if  $B_i h^L < 1.$

\end{corollary}
We note that in the typical case where  $B^K$ has full rank $q<k$,
$B^K h^L = B^K h^K$ implies $h^L = h^K$, so the the formulations of Theorem \ref{thm:multi} and Corollary \ref{cor:alt} are equivalent.  

}


Unlike Theorem \ref{thm:1}, Theorem \ref{thm:multi} does not provide a direct computational method for determining $w^L$. Instead, it provides a necessary condition satisfied by the active assets in terms of their $q$-vectors $B_i$, $i=1,\dots,p$, as a generalization of the 1-factor setting: the included $B_i$ are the ones below a certain threshold, where the threshold in $q$ dimensions is determined by a $(q-1)$-dimensional hyperplane $H$.

{
In typical applications, such as in the illustration below in section \ref{sec:2factorexample}, $p$ is much greater than the number $n$ of samples. In this case, when the factor exposures are bounded and have a positive variance across assets, we expect that $B^K (\Delta^K)^{-1}B^K$ will dominate $\Omega^{-1}$ by a factor proportional to $p$. In the limiting case where $\Omega^{-1}$ is set to zero, $h^L$ in equation \eqref{eq:sephyp} can be viewed as a vector of coefficients of a weighted least squares regression of $\bfo$ onto the columns of $B^K$ with weight matrix $\Delta^{-1}$.\footnote{We thank Lisa Goldberg for this observation.}
}

In the one-factor case $q=1$, the condition of Theorem \ref{thm:multi} reduces
to \eqref{eq:thresholdL}: in this case, $B$ is a single column vector corresponding to $\beta$ in the single index model, and $\Omega$ is a scalar $\sigma^2$. A computation using the Woodbury identity for the inverse of $\Sigma^K$,
\begin{eqnarray}
    (\Sigma^K)^{-1} = (\Delta^K)^{-1}\{I - B^K\left[\frac{1}{\sigma^2} + (B^K)^\top (\Delta^K)^{-1} B^K\right]^{-1}(B^K)^\top (\Delta^K)^{-1}\},
\end{eqnarray}
shows that the scalar $h^K = \sigma^2 (B^K)^\top (\Sigma^K)^{-1} \bfo_k$ is given by
\begin{equation}
   h^K = \frac{\sum_{j \in K} \frac{B_j}{\delta_j^2}}{\frac{1}{\sigma^2} + \sum_{j \in K} \frac{B_j^2}{\delta_j^2}}
\end{equation}
and the condition $B_i h^K < 1$ is equivalent to \eqref{eq:thresholdL}.

\section{Numerical Examples}

\subsection{Single factor exposure estimation} \label{subsec:beta_est}

Given observed excess returns of $p$ assets over $n$ observation periods, we may form the $p \times n$ data matrix $Y$ and the resulting sample covariance matrix $S = \frac{1}{n} YY^\top$.

In the likely event that $p > n$, the matrix $S$ will be rank deficient, so we need a model for estimating an invertible covariance matrix for use in portfolio optimization.

A standard approach is to use factor models to accomplish this.  In this section we  consider a one-factor model of returns,
\begin{equation} \label{eq:return-model}
    r = \beta f + \epsilon
\end{equation}
where $\beta$ is an unknown vector of factor exposures, $f$ is a random variable with mean zero and variance $\sigma^2$ representing the factor return, and $\epsilon$ is a random vector whose entries are mutually independent and independent of $f$, with mean zero and covariance $\Delta$. The columns of $Y$ are then assumed to be $n$ independent realizations of $r$.  The covariance matrix of $r$ is given by
\eqref{eq:singlefactor}:
\begin{equation*}
    \Sigma = \sigma^2 \beta \beta^\top + \Delta.
\end{equation*}

Setting $\zeta^2 = |\beta|^2 \sigma^2$ and $b = \beta/|\beta|$, the one-factor covariance model may be written
\begin{equation}\label{eq:pop_cov}
    \Sigma = \zeta^2 bb^\top + \Delta,
\end{equation}
where recall $\Delta = diag(\delta_1^2,\dots,\delta_p^2)$.
Only $Y$ and the corresponding sample covariance matrix $S$ are observed. Estimating $\Sigma$ requires estimating the scalar $\zeta^2$, the unit vector $b$, and the vector $(\delta_1^2,\dots,\delta_p^2)$ of idiosyncratic variances.

For the purpose of computing a minimum variance portfolio, we will consider three different data-driven estimators of \eqref{eq:pop_cov}: $\sjse, \smjse$, and $\sms$, defined as follows.

The estimator $\sjse$ is a Bayesian-style James-Stein shrinkage estimator
 with the advantage that the shrinkage takes place entirely inside the class of single-factor models:
\begin{equation} \label{eq: jse}
    \sjse = \eta^2 \hjse (\hjse)^\top + \djse ,
\end{equation}
where $\eta^2, \hjse, \djse$ are estimators of $\zeta^2, b, \Delta$ described further in Section \ref{sec:beta_est}.
The JSE estimators are designed for large $p >> n$ and correct for high-dimensional statistical bias in the sample eigenvectors.  Suitable parameters are $p=1000$, $n = 126$ as in the empirical illustrations in the next section.
This is a single-factor model where the normalized factor loadings $\hjse$ are asymptotically good estimates of the unknown population factor exposure unit vector $b = \beta/|\beta|$ responsible for generating the observed returns via \eqref{eq:return-model}.

The remaining two estimators $\smjse, \sms$ are single-index market models in the manner of \cite{sharpe1963} and as used by \cite{cst2011}.

Given any data-driven estimator $\hat \Sigma$ of $\Sigma$, and cap-weighted market portfolio $w_M$, we may define the estimated market variance and market beta by
\begin{equation} \label{eq:mkt_sigma}
    \sigma_M^2 = w_M^\top \hat \Sigma w_M
\end{equation}
and 
\begin{equation} \label{eq:mkt_beta}
    \beta^M = \frac{\hat \Sigma w_M}{\sigma_M^2}.
\end{equation}

From these, we form a single index market covariance matrix $\Sigma^{M,{\hat \Sigma}}$ depending on $w_M$ and $\hat \Sigma$:
\begin{equation} \label{eq:mmodel}
    \Sigma^{M,{\hat \Sigma}} = \sigma_M^2 \beta^M (\beta^M)^\top + \Delta^M ,
\end{equation}
where
\begin{equation} \label{eq:delta-m}
  \Delta^M = diag((\delta^M_i)^2), \quad  (\delta^M_i)^2 = S_{ii} - (\beta^M_i)^2 \sigma_M^2, 
\end{equation}
or, equivalently,
\begin{equation}
    \Sigma^{M,{\hat \Sigma}} = \zeta_M^2 b^M(b^M)^\top + \Delta^M
\end{equation}
where
\begin{equation}
    \zeta_M^2  = \sigma_M^2 |\beta^M|^2, \quad b^M = \beta^M/|\beta^M|.
\end{equation}

We examine the choices $\hat \Sigma = \sjse$ and $\hat \Sigma = S$, and define
\begin{eqnarray}
    \smjse &=& \Sigma^{M,{\sjse}} \\
    \sms &=& \Sigma^{M,S}.
\end{eqnarray}

It can be shown that $\smjse$ and $\sjse$ as asymptotically consistent (as $p \to \infty$) in a sense described in Section \ref{subsec:consistency} below.



\medskip
\subsection{Empirical example for a single index model} \label{sec:empirical}
In this section we apply our results in an empirical illustration using daily returns of the top $p=1000$ US stocks by market capitalization for the period January 3, 2022 to July 1, 2022.\footnote{Returns were gathered from the CRSP database via the Wharton Research Data Service, then cleaned and centered before use. We removed one smaller stock of a firm with artificially low volatility due to an imminent merger, and added the next largest asset to keep the total at 1,000.}
 
From the market capitalizations $mc(i), i=1,\dots,p$, at the beginning of the period, we determine the market portfolio $w_M$ by
\begin{equation}
    w_M(i) = \frac{mc(i)}{\sum_{i=1}^p mc(i)}.
\end{equation}
We then use the formulas of section \ref{subsec:beta_est} to determine three choices\footnote{\cite{cst2011} use a different choice of single-index market covariance matrix by taking $\hat \Sigma$ to be a Ledoit-Wolf shrinkage estimator, see \cite{ledw2004}. The empirical outcomes are similar to the ones reported here.}
of the data-driven covariance estimator for computing the LOMV portfolio: $\sjse$, $\smjse$, and
$\sms$.

The outcomes for the resulting long only portfolio size, market factor variance, and estimated betas are summarized in Table 1.
For $\sjse = \eta^2 \hjse {\hjse}^\top + \djse$, the market quantities $\sigma_M^2$ and $\beta^M$ are undefined and $\hjse$ is scaled as a unit vector by convention.


\begin{table}[H] \label{table:three_cases}
\begin{center}
\begin{tabular}{c || c c c}
estimator & $k$ & $\sigma^2_M$(\%) & mean$(\beta^M)$ \\
\hline 
$\smjse $&65 & 5.8 & 1.08 \\
$\sms$&51 & 6.2 & 1.13 \\
$\sjse$&66 & * & * \\
\end{tabular}
\end{center}
\caption{Outcomes for each of three choices of covariance estimator. The value $k$ is the number of active assets in the long-only minimum variance portfolio (out of 1000). $\sigma^2_M$ and $\beta^M$ are the derived parameters according to equations \eqref{eq:mkt_sigma} and \eqref{eq:mkt_beta} but does not apply to the last row.}
\end{table}

Table 2 shows that the three methods mostly agree on which active assets should be included in the long-only minimum variance portfolio.

\begin{table}[H] \label{table:common_assets}
\begin{center}
\begin{tabular}{c || c c c c c c c c c}
portfolio & MJSE  & MS &JSE  & MJSE $\cap$ MS & MJSE $\cap$ JSE & JSE $\cap$ MS \\
\hline 
\# assets & 65 & 51 & 66 & 49 &65 &49 \\
\end{tabular}
\end{center}
\caption{The number of assets that each of the three estimated minimum variance portfolios have in common, out of 1000 total assets considered. }
\end{table}

Figures \ref{subfig:mjse-v-jse} and \ref{subfig:mjse-v-ms} show scatterplots demonstrating that the portfolios for the market model $\smjse$ and the statistical factor model $\sjse$ are almost the same in this experiment, but the $\sms$ portfolio noticeably differs.

\begin{figure}[H]
\centering
\begin{subfigure}{0.48\textwidth}
\centering
\includegraphics[width=\linewidth]{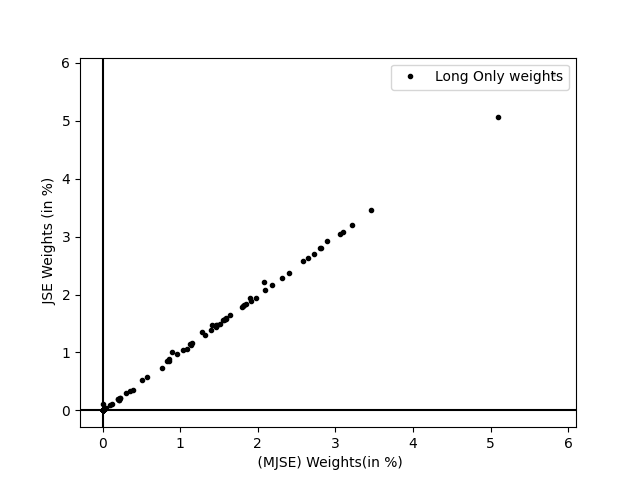}
\caption{Weights of MJSE vs JSE portfolios.}
\label{subfig:mjse-v-jse}
\end{subfigure}%
\begin{subfigure}{0.48\textwidth}
\centering
\includegraphics[width=\linewidth]{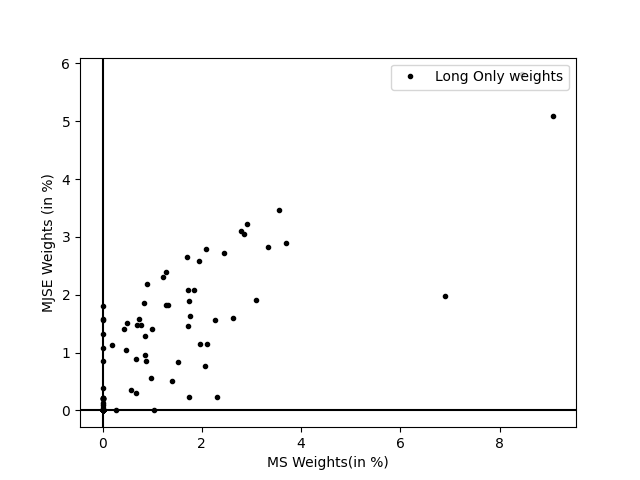}
\caption{Weights of MS vs MJSE portfolios.
}
\label{subfig:mjse-v-ms}
\end{subfigure}
\caption{Comparison of portfolio asset weights for the three portfolios MJSE, JSE, and MS, plotted in percent.}
\label{fig:double}
\end{figure}

Figures \ref{fig:w-v-beta}, \ref{fig:w-v-delta}, and \ref{fig:delta-v-beta} below illustrate the relationships between estimated market beta, portfolio weight, and idiosyncratic risk for the MSJE portfolio.  (Plots look similar for the other portfolios.)

The betas range between $-0.05$ and $3.62$ with the only negative beta being approximately $-0.059$. The maximum beta in the long-only portfolio is 0.281. The idiosyncratic risk ranges from $0$ to $46\%$. 
Figures \ref{fig:w-v-beta} and \ref{fig:w-v-delta} show the 1000 individual asset weights under the single-factor model for both the long-short (blue dots) and the long-only (black dots) minimum-variance portfolio plotted against market beta and delta.  The weights of the active assets in the long-only portfolio had a maximum value of $5.81\%$. 

Of interest is the fact that only 65 of the 1000 securities were active in the long-only portfolio. 
  In addition to having lower beta, assets with low idiosyncratic risk are more likely to be active in the long-only portfolio. We also see that assets with lower idiosyncratic risk tend to have higher absolute value weights in the long-short portfolio.

\begin{figure}[H]
\centering    
    \includegraphics[width=.7\textwidth]{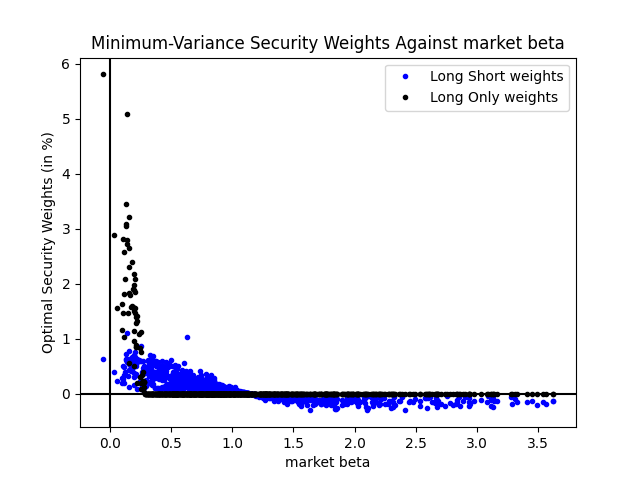}\hfill

    \caption{MJSE portfolio weights against market beta for the top 1000 US stocks, estimated for daily returns in the first half of 2022.}\label{fig:w-v-beta}
\end{figure}

\begin{figure}[H]
 \centering  
    \includegraphics[width=.7\textwidth]{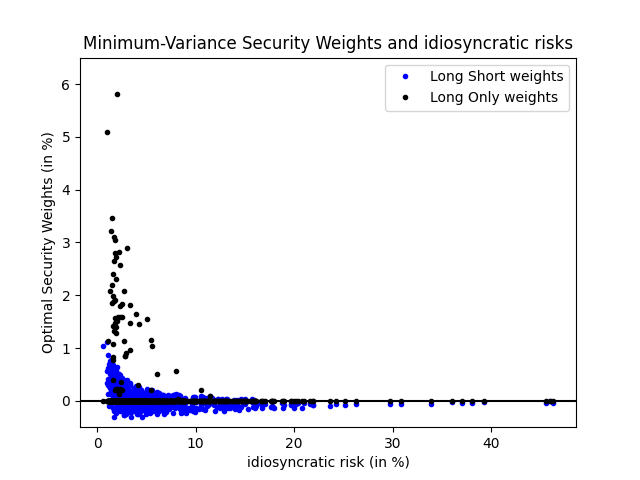}\hfill

    \caption{MJSE portfolio weights against specific risk for the top 1000 US stocks, from daily returns in the first half of 2022.}\label{fig:w-v-delta}
\end{figure}

\begin{figure}[H]
\centering   
    \includegraphics[width=.7\textwidth]{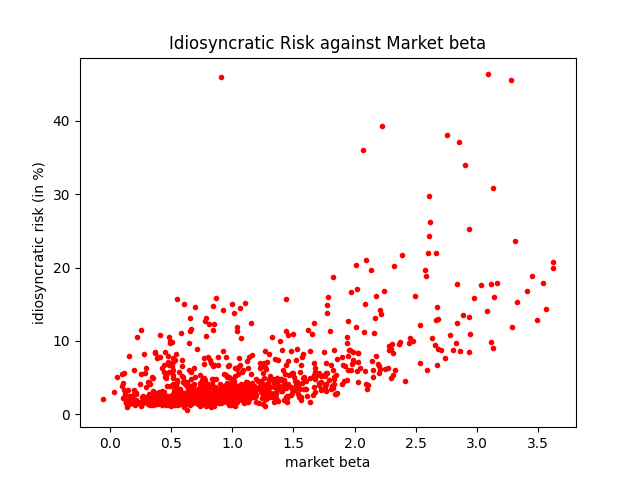}\hfill

    \caption{MJSE market beta against specific risk for the top 1000 US stocks in the first half of 2022.}\label{fig:delta-v-beta}
\end{figure}

\subsection{Empirical example for a two-factor model} \label{sec:2factorexample}

To illustrate the hyperplane separation of factor loadings described in Theorem \ref{thm:multi}, we fit a two-factor model ($q=2$) to the same daily returns of the 1,000 stocks used in Section \ref{sec:empirical}. Starting from
the leading two eigenvalue-eigenvector pairs of the sample covariance matrix,
as described further in Section \ref{subsec:JSM}, we apply the James-Stein-Markowitz (JSM) multifactor shrinkage method of \cite{skggb2024} to obtain the covariance estimate
\begin{equation}
    \sjsm = B \Omega B^\top + \Delta ,
\end{equation}
 where $B$ is a $2 \times p$ matrix with orthogonal columns, obtained via shrinkage from the leading two sample eigenvectors; $\Delta$ is a diagonal matrix of specific variances; and  $\Omega$ is a $2 \times 2$ diagonal matrix. 
 The row $B_i$ of asset $i$ is that asset's exposure vector to the two factors.  The resulting long-only minimum variance portfolio is computed by means of the open source convex optimization package cvxpy (www.cvxpy.org).

The results are displayed in the left hand plot of Figure \ref{fig:separation} below, in which the horizontal axis shows the exposure to the leading (market) factor (the first component of $B_i$), and the vertical axis to the second factor.   Orange indicates assets included in the LOMV portfolio based on the model, and blue indicates excluded assets. The solid red line is the hyperplane of Theorem \ref{thm:multi} separating the active and inactive assets in the LOMV portfolio, while the dotted red line indicates the single-index model threshold value that would apply for the same data with a one-factor model.  We observe a relatively small number of assets in the LOMV portfolio.

The right-hand plot in Figure \ref{fig:separation} shows a close-up of the LOMV portfolio's active assets with the portfolio weights coded as marker size and color.  The largest weight is 10\%.  The portfolio weight of an asset is proportional\footnote{With a proportionality constant that is the same for all assets. This follows from a computation using the Woodbury identity.}
to the ratio $d_i/\delta_i^2$, where $\delta_i^2$ is the specific variance of asset $i$, and $d_i$ is the perpendicular distance from the asset's beta point $B_i$ to the separating hyperplane.  {Although the active asset portfolio weights shown in Figure \ref{fig:separation} generally vary inversely with the distance from the separating hyperplane, unusually large or small specific variances can play a dominant role, as is the case for three assets with large exposure to the second factor.}

\begin{figure}[H]
 \centering  
    \includegraphics[width=\textwidth]{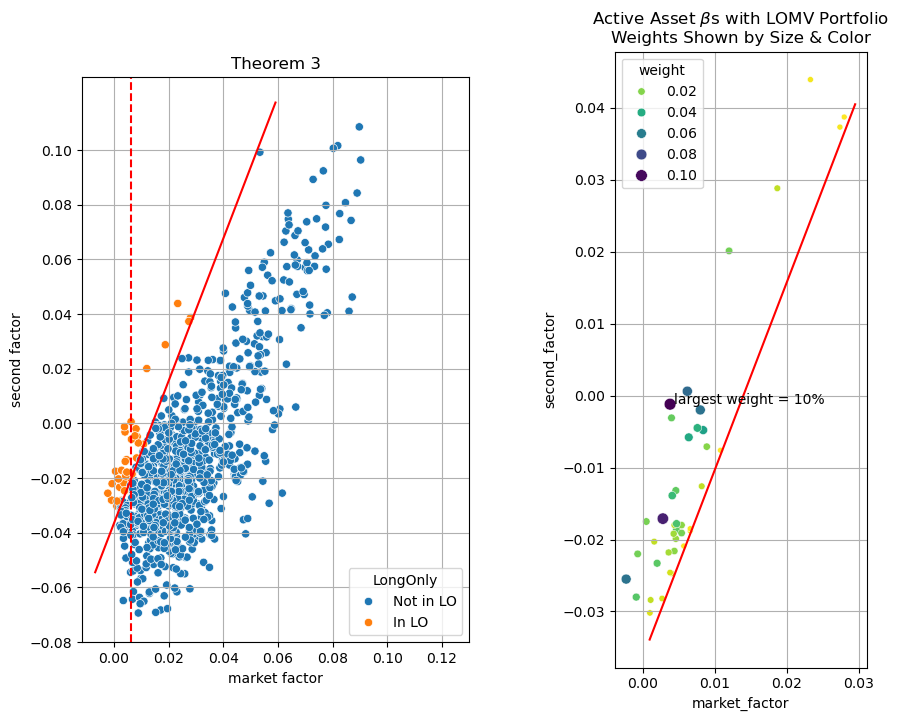}\hfill

    \caption{2-vector exposures for each asset from a statistical 2-factor JSM return covariance matrix, plotted for each of the 1000 assets. On the right is a close-up showing portfolio weights.}\label{fig:separation}
\end{figure}

An interesting comparison with the portfolio resulting from a single-index model is evident by examining the vertical dotted red line {in the left-hand plot}, which, {as stated above}, is the one-factor threshold for the same returns data. The 1-factor portfolio excludes a few of the assets with higher exposure to both factors and includes significantly more with {negative} exposure to the second factor.  For the two-factor model portfolio, exposure to the second factor can {compensate} for a {fairly} significant positive exposure to the 
first factor, such as for the two {uppermost} orange assets {in left-hand plot}. These assets are significantly into the excluded region defined by the dotted line.  Overall, the 2-factor model includes {41} assets, {24} fewer than the 65 assets of one-factor model.

The angle between the two lines could be considered a measure of the importance of the second factor in determining active long-only assets. The angle here is 21.5\% of a right angle, which is relatively significant.

\section{Proofs} \label{sec:proof}

\subsection{Proof of Theorem \ref{thm:wL}}

Notation: $\bfo_p$ is the $p$-dimensional vectors of all ones, and similarly $\bfz_p$ for all zeros.

The solution $w=w^L$ of our constrained optimization problem \eqref{eq: prob 1} satisfies the well-known Karush-Kuhn-Tucker (KKT) conditions\footnote{See, for example, \cite{boyd-vandenberghe2004} or \cite{beck2023}.}, which are a set of equations in $w$, an auxiliary $p$-vector ${\lambda}$, and an auxiliary scalar $\nu$ (the Lagrange multipliers) as follows:

    \begin{align}
           &2\Sigma w- {\lambda} +\nu \bfo_p=\bfz_p \label{eq:KKT}\\ 
           & w^\top \bfo_p  =1 \label{eq:KKTb}\\
           &\lambda_i w_i =0 \text{ }i=1,2,...,p \label{eq:KKTc}\\ 
           &\lambda_i \geq 0, w_i \geq 0 \text{ } i=1,2,...,p. \label{eq:KKTd}
    \end{align}

These KKT conditions include $2p+1$  equations \eqref{eq:KKT} -- \eqref{eq:KKTc} in the $2p+1$ unknowns $$w_1,\dots,w_p,\lambda_1,\dots,\lambda_p, \nu,$$ along
with $2p$ inequality constraints \eqref{eq:KKTd}.

Define the (necessarily non-empty) set
\begin{equation}
    K = \{ i \leq p: w_i > 0 \}.
\end{equation}

Let $k >0$ denote the cardinality of $K$. For any vector $x \in \mathbf{R}^p$, denote by $x^K \in \mathbf{R}^k$ the $k$-dimensional vector obtained from $x$ by deleting the entries $x_j$ for all $j \notin K$.
Likewise, for any $p \times p$ matrix $M$, denote by $M^K$ the $k \times k$ principal submatrix obtained from $M$ by deleting all the rows and columns with indices outside $K$.


Since $\Sigma$ is symmetric positive definite, so is the submatrix $\Sigma^K$, and hence $\Sigma^K$ is invertible. 
Further,
\begin{equation}
    w^\top \Sigma w = (w^K)^\top \Sigma^K w^K.
\end{equation}

Since $\lambda_i = 0$ for all $i \in K$, taking the $k$ rows of equation \eqref{eq:KKT} corresponding to indices in $K$ tells us
\begin{equation} \label{eq:ellKKT}
    2\Sigma^K w^K + \nu \bfo_k = \bfz_k,
\end{equation}
where here the vectors $\bfo_k$ and $\bfz_k$ are $k$-dimensional.  Multiplying \eqref{eq:ellKKT} on the left by $\bfo_k^\top (\Sigma^K)^{-1}$ and using $\bfo_k^\top w^K = 1$, we obtain
\begin{equation} \label{eq:nu}
    \nu = \frac{-2}{\bfo_k^\top (\Sigma^K)^{-1} \bfo_k}
\end{equation}
and therefore
\begin{equation} \label{eq:w*}
    w^K = \frac{(\Sigma^K)^{-1}\bfo_k}{\bfo_k^\top (\Sigma^K)^{-1}\bfo_k} ,
\end{equation}
or, equivalently,
\begin{equation}
    w^L  = \frac{(\Sigma^{K,0})^+\bfo_p}{\bfo_p^\top(\Sigma^{K,0})^+\bfo_p} .
\end{equation}

\subsection{Proof of Theorem \ref{thm:1}}

Part 1 of the theorem follows solely from the monotonicity of the sequence $\{\beta_i\}$.
For convenience, define
\begin{equation}
    C_i = \sum_{j=1}^{i}\frac{\beta_j}{\delta_j^2}.
\end{equation} 
It is easy to verify, for all $i= 1,\dots, p-1$, that
\begin{equation}
  R_{i+1} - R_i = (-\beta_{i+1} + \beta_i) C_i.    
\end{equation}
 Since $(-\beta_{i+1} + \beta_i)$ is always non-positive, $R_{i+1} - R_i \leq 0$ when $C_i \geq 0$ and $R_{i+1} - R_i \geq 0$ when $C_i \leq 0$.

Let $s = \min\{i : C_i > 0\}$, $1 \leq s \leq p$.
If $j < s$, then $C_j \leq 0$ by definition of $s$, so $R_{j+1}-R_j \geq 0$. Since $\beta_s >0$, $C_i$ is increasing for $i \geq s$, so if $j \geq s$ then $C_j > 0$, and we have $R_{j+1} - R_j \leq 0$.  This establishes the conclusion of part 1.
\medskip

We proceed to Part 2. 
Our goal now is to determine $K = \{i \leq p : w^L_i > 0 \}$ in explicit form in term of the parameters $\sigma, \beta, \delta$ of the problem.

Let $k = |K|$ and define
\begin{equation}
    \ell = \max \{i \leq p: R_i >0 \}.
\end{equation}
We will complete the proof by establishing
$$ k = \ell  \mbox{ and } K = \{1,2,\dots, k\}.$$

Let $\beta^K \in \mathbf{R}^k$ and the $k \times k$ matrix $\Sigma^K$ be determined from $K$ as before, obtained from $\beta$ and $\Sigma$ by deleting rows or rows and columns corresponding to indices outside $K$.

By the Woodbury identity, 
$$(\Sigma^K)^{-1}= diag(\frac{1}{(\delta^K)^2})-\frac{(\frac{\beta^K}{(\delta^K)^2})(\frac{\beta^K}{(\delta^K)^{2}})^\top}{\frac{1}{\sigma^2}+ (\frac{\beta^K}{(\delta^K)^2})^\top\beta^K}
$$
where $\frac{1}{({\delta^K})^2} = [\frac{1}{\delta_j^2}: j \in K]^\top$ and
$\frac{\beta^K}{(\delta^K)^{2}}= [\frac{\beta_j}{\delta_j^2}: j\in K]^\top$.

This means
\begin{equation}
    (\Sigma^K)^{-1}\bfo_{k} = \frac{1}{(\delta^K)^2} - \frac{(\frac{\beta^K}{(\delta^K)^2})(\frac{\beta^K}{(\delta^K)^{2}})^\top \bfo_{k}}{\frac{1}{\sigma^2}+ (\frac{\beta^K}{(\delta^K)^2})^\top\beta^K}.
\end{equation}
Now if $i \in K$, then   
\begin{equation} \label{eq:vi}
    ((\Sigma^K)^{-1}\bfo_{k})_i = \frac{1}{\delta_i^2} - \frac{\frac{\beta_i}{\delta_i^2}\sum_{j \in K}\frac{\beta_j}{\delta_j^2}}{\frac{1}{\sigma^2}+\sum_{j\in K} \frac{\beta^2_j}{\delta_j^2}} >0.
\end{equation}
Clearing the positive denominators, we
obtain
\begin{equation} \label{eq:beta-condition}
 \frac{1}{\sigma^2}+\sum_{j\in K} \frac{\beta^2_j}{\delta_j^2}- \beta_i\sum_{j\in K}\frac{\beta_j}{\delta_j^2} > 0   
\end{equation}
or
\begin{equation} \label{eq:BK>CK}
    B_K > \beta_i C_K
\end{equation}
where 
\begin{equation}
 B_K = \frac{1}{\sigma^2}+\sum_{j\in K} \frac{\beta^2_j}{\delta_j^2}, \quad C_K = \sum_{j\in K}\frac{\beta_j}{\delta_j^2}.  
\end{equation}

Now suppose instead $i \notin K$, so $w_i=0$. 
Since $\Sigma w = \sigma^2 \beta \beta^\top w + diag(\delta^2)w$ and $w_j=0$ for all $j \notin K$,
we have
\begin{equation} \label{eq:sigmawi}
    (\Sigma w)_i = \sigma^2 \beta_i \sum_{j \in K} \beta_j w_j.
\end{equation}
Conditions \eqref{eq:KKT} and \eqref{eq:KKTd} tell us
\begin{equation}
    0 \leq \lambda_i = 2(\Sigma w)_i + \nu,
\end{equation}
or, using \eqref{eq:sigmawi},
\begin{equation}
    0 \leq 2 \sigma^2 \beta_i \sum_{j \in K} \beta_j w_j + \nu.
\end{equation}
For $j \in K$, we have, from \eqref{eq:w*},
\begin{equation}
    w_j = \frac{((\Sigma^K)^{-1}\bfo_{k})_j}{\bfo_{k}(\Sigma^K)^{-1}\bfo_{k}}.
\end{equation}
Using this, substituting for $\nu$ with \eqref{eq:nu}, and multiplying through by $\bfo_{k}(\Sigma^K)^{-1}\bfo_{k}/2$ gives us
\begin{equation}
    0 \leq \sigma^2 \beta_i \sum_{j \in K} \beta_j ((\Sigma^K)^{-1}\bfo_{k})_j - 1.
\end{equation}
By \eqref{eq:vi}, 
\begin{equation}
    ((\Sigma^K)^{-1}\bfo_{k})_j = \frac{1}{\delta_j^2} - \frac{\beta_j}{\delta_j^2} \frac{C_K}{B_K}.
\end{equation}
Using this in the previous inequality, we have
\begin{eqnarray}
   0 &\leq& \sigma^2 \beta_i \sum_{j \in K} \frac{\beta_j}{\delta_j^2} (1 - \frac{\beta_j C_K}{B_K}) -1 \\
   &=& \sigma^2 \beta_i \big( C_K - \frac{C_K}{B_K} (B_K - \frac{1}{\sigma^2}) \big) -1 \\
   &=& \sigma^2 \beta_i \frac{C_K}{\sigma^2 B_K} -1 = \beta_i \frac{C_K}{B_K} -1,
\end{eqnarray}
or
\begin{equation} \label{eq:BK<CK}
    B_K \leq \beta_i C_K.
\end{equation}

Combining \eqref{eq:BK>CK} and \eqref{eq:BK<CK}, we have established the following
\begin{lemma} \label{lem:dichotomy}
With $B_K$ and $C_K$ as defined above,
$i \in K$ if and only if $B_K > \beta_i C_K$.
\end{lemma}
\begin{corollary}
    $C_K \neq 0$.
\end{corollary}
{\em Proof of Corollary.}
If ${k} = p$, then $C_K \neq 0$ is our standing assumption. Otherwise there exists $j \notin K$, so by Lemma \ref{lem:dichotomy} $B_K \leq \beta_j C_K$.
But $B_K >0$, so again $C_K$ cannot be zero.
\qed

\begin{lemma} \label{lem:segment}
    Recall ${k}$ is the cardinality of $K$.
If $C_K >0$, then $K = \{1,2,\dots, {k} \}$. If $C_K < 0$, then $K = \{p-{k}+1,\dots,p\}$.
\end{lemma}

{\em Proof of lemma.} Recall that the $\beta_i$ are arranged in increasing order. For all $i \in K$ and $j \notin K$, we have
\begin{equation}
    \beta_i C_K < B_K \leq \beta_j C_K,
\end{equation}
hence
\begin{equation}
    \beta_i C_K < \beta_j C_K.
\end{equation}
If $C_K >0$, this means that $\beta_i < \beta_j$ for all $i\in K, j\notin K$, and hence
$i < j$ for all $i \in K, j \notin K$. Therefore $K$ must be an initial segment of the sequence $\{1,2,\dots,p\}$.

Similarly, if $C_K <0$, then the reverse inequality is true, and $K$ must be a terminal segment of $\{1,2,\dots,p\}$. \qed

Next, define
\begin{equation}
 C_P = \sum_{j=1}^p\frac{\beta_j}{\delta_j^2}, 
\end{equation}
and recall $C_P > 0$ by our standing assumption.

\begin{lemma} \label{lem:CKpos}
   $C_K > 0$.
\end{lemma}

{\em Proof.}  If ${k} = p$ then $C_K = C_P$ and there is nothing to prove, so consider the case ${k} < p$.  We know that $C_K \neq 0$. Suppose for
contradiction that $C_K < 0$.  By Lemma \ref{lem:segment},
this means that $K = \{p-{k}+1,\dots, p\}$.

Note
\begin{equation} \label{eq:CP}
    C_P = \sum_{j=1}^{p-{k}} \frac{\beta_j}{\delta_j^2} + C_K.
\end{equation}
Now $0 > C_K = \sum_{j=p-{k}+1}^{p} \frac{\beta_j}{\delta_j^2}$, so at least one of the terms of this sum must be negative.  Then, since the beta sequence is increasing, $\beta_{p-{k}+1} <0$ and  $\beta_j < 0$ for all $j \leq p - {k}$, and hence
\begin{equation}
    \sum_{j=1}^{p-{k}} \frac{\beta_j}{\delta_j^2}  < 0.
\end{equation}
This forces $C_P < 0$ via \eqref{eq:CP}, a contradiction. Hence we must have $C_K>0$ and $C_K = \sum_{j=1}^{k} \frac{\beta_j}{\delta_j^2}$.

\qed

\begin{lemma} \label{lem:threshold}
    Under our standing assumption, for the long-only solution $w$ of problem \eqref{eq: prob 1}, 
\begin{equation}
     w_i > 0 \mbox{ if and only if }   \beta_i < \frac{B_K}{C_K} =  \frac{\frac{1}{\sigma^2}+\sum_{j=1}^{k} \frac{\beta^2_j}{\delta_j^2}}{\sum_{j=1}^{k}\frac{\beta_j}{\delta_j^2}}.
\end{equation}
\end{lemma}

{\em Proof.} From Lemma \ref{lem:CKpos}, $C_K >0$.
By Lemma \ref{lem:segment}, $K = \{1,2,\dots, {k}\}$.
The result follows from Lemma \ref{lem:dichotomy}. 
\qed
\medskip

By Lemmas \ref{lem:segment} and \ref{lem:CKpos}, we have established that
\begin{equation}
    K = \{1,2,\dots,{k} \}.
\end{equation}
It remains to show that ${k} = \ell$. 
Recall $R_1 = 1/\sigma^2$ and, for $i=2,\dots,p$,
\begin{equation} \label{eq:Rdef}
 R_i = {\frac{1}{\sigma^2} + \sum_{j=1}^{i-1}\frac{\beta_j}{\delta_j^2}(\beta_j - \beta_i)}.   
\end{equation}

Now, by Lemma \ref{lem:threshold},  $\beta_{k} < B_K/C_K \leq \beta_{{k}+1}$.
Also,
\begin{eqnarray}
    R_{{k}} &=& {\frac{1}{\sigma^2} + \sum_{j=1}^{{k}}\frac{\beta_j^2}{\delta_j^2} - \beta_{{k}}\sum_{j=1}^{{k}}\frac{\beta_j}{\delta_j^2}} 
    \\
    &=&  B_K - \beta_{{k}}C_K > 0.
\end{eqnarray}
Likewise
\begin{eqnarray}
    R_{{k}+1} =  B_K - \beta_{{k}+1}C_K \leq 0.
\end{eqnarray}

By part 1 of the Theorem, the sequence $\{R_i\}$ crosses zero at most once, so this establishes
\begin{equation}
    {k} = \max\{i \leq p: R_i > 0 \} = \ell,
\end{equation}
completing the proof of Part 2.
\qed

\subsection{Proof of Theorem \ref{thm:multi}}

Let $w = w^L$ be the solution of the long-only problem for
\begin{equation}
    \Sigma = B\Omega B^\top + \Delta,
\end{equation}
and $K = \{i \leq p: w^L_i >0 \}$,
with $\ell = |K|$.
For any $p$-vector $v$, for purposes of this proof we adopt the notation that $v^K$ as defined before, and $v'$ is the complementary $(p-\ell)$-vector obtained from $v$ by deleting all the coordinates in $K$; $B'$ is obtained by deleting all the rows labeled by entries in $K$, and $B^K$ by deleting the rows not in $K$. 

First recall that $i \notin K$ implies $w_i =0$, and hence
\begin{align}
    (\Sigma w)' = \left((B\Omega B^\top +\Delta)w\right)'  =  (B\Omega B^\top w)' =  B' \Omega(B^K)^\top w^K . \label{eq:sigmaw}
\end{align}

Recall equations \eqref{eq:KKT} and \eqref{eq:KKTd}:
\begin{equation*}
    2\Sigma w - \lambda + \nu \bfo_p =\bfz_p; \quad
    \lambda_i \geq 0, w_i \geq 0, \, i=1,\dots,p.
\end{equation*}
Therefore
\begin{equation}
    \bfz_{p-\ell} \leq \lambda'  =  2(\Sigma w)' + \nu \bfo_{p-\ell} 
\end{equation}
and hence
\begin{equation}\label{eq:KKT_1}
    \bfz_{p-\ell} \leq 2B' \Omega(B^K)^\top w^K + \nu \bfo_{p-\ell}
\end{equation}
using equation \eqref{eq:sigmaw}.

Substituting using \eqref{eq:nu} and \eqref{eq:w*}, clearing denominators, and dividing by 2, we obtain
\begin{align}
    \bfz_{p-\ell} \leq B'\Omega(B^K)^\top (\Sigma^K)^{-1}\bfo_\ell -\bfo_{p-\ell} = B' h^K - \bfo_{p-\ell}
\end{align}
with 
$$h^K = \Omega (B^K)^\top (\Sigma^K)^{-1} \bfo_\ell$$
 as defined in \eqref{eq:hk} of
 Theorem \ref{thm:multi}.

Thus we have established that for every $i \notin K$,
\begin{equation}
    1 \leq B_i h^K.
\end{equation}

Conversely, recall
\begin{equation}
  \Sigma^K = B^K \Omega (B^K)^\top + \Delta^K.  
\end{equation}
We therefore have
\begin{equation}
    B^K h^K = B^K \Omega (B^K)^\top (\Sigma^K)^{-1} \bfo_\ell
    = (\Sigma^K - \Delta^K)(\Sigma^K)^{-1} \bfo_\ell =
    \bfo_\ell - \Delta^K(\Sigma^K)^{-1} \bfo_\ell .
\end{equation}
From equation \eqref{eq:w*},
$
    (\Sigma^K)^{-1} \bfo_\ell = (\bfo_\ell^\top (\Sigma^K)^{-1} \bfo_\ell) w^K 
$, and hence
\begin{equation}
    B^K h^K = \bfo_\ell - (\bfo_\ell^\top (\Sigma^K)^{-1} \bfo_\ell) \Delta^K w^K < \bfo_\ell,
\end{equation}
since the entries of
$(\bfo_\ell^\top (\Sigma^K)^{-1} \bfo_\ell) \Delta^K w^K$ are all positive.
We conclude that
\begin{equation}
    B_i h^K < 1
\end{equation}
for all $i \in K$.
\qed

\section{Single factor beta estimation} \label{sec:beta_est}

\subsection{The JSE estimator}

Covariance matrix shrinkage estimators, e.g. \cite{ledw2004}, have enjoyed wide adoption for various problems when estimating large-dimensional covariance matrices from data, and typically take the form
\begin{equation}
    \Sigma = \alpha S + (1-\alpha)T,
\end{equation}
where $S$ is a sample covariance matrix, $T$ is a suitable target matrix, such as a scalar matrix, and $\alpha \in (0,1)$ is defined to minimize some error function.

However, in cases where we are working within the structure of a factor model, we are primarily interested in estimating the leading covariance eigenvector(s), from which an estimated factor model can be defined.  Eigenvector shrinkage is the subject of a recent stream of research in \cite{goldberg2020, goldberg2023, goldberg2022,strat-spec2025, shkolnik2022}. We call it JSE (James-Stein for Eigenvectors) because the shrinkage formulas themselves turn out to be close analogs of the classical James-Stein shrinkage formulas for estimation of multivariate means. The theory provides asymptotic improvements in the HL asymptotic regime corresponding to the limit as the dimension $p \to \infty$ with the number of samples $n$ fixed.

In this section we describe how to compute the JSE estimator of the leading eigenvector.

As in Section \ref{subsec:beta_est}, we  consider a single factor model of returns,
\begin{equation}
    r = \beta f + \epsilon .
\end{equation}
  The covariance matrix of $r$ is given by
\eqref{eq:singlefactor}:
\begin{equation*}
    \Sigma = \sigma^2 \beta \beta^T + \Delta.
\end{equation*}

Setting $\zeta^2 = ||\beta||^2 \sigma^2$ and $b = \beta/||\beta||$, the single index model may be written
\begin{equation}
    \Sigma = \zeta^2 bb^\top + \Delta,
\end{equation}
where recall $\Delta = diag(\delta_1^2,\dots,\delta_p^2)$.
 Estimating $\Sigma$ requires estimating the scalar $\zeta^2$, the unit vector $b$, and the vector $(\delta_1^2,\dots,\delta_p^2)$ of idiosyncratic variances.

Suppose we observe a time series of $n$ returns of each of $p$ assets, and $p >> n$ as might typically be the case when $n$ is limited by non-stationarity or data availability. The observations can be summarized by a $p \times n$ data matrix $Y$, and the resulting sample covariance matrix is
\begin{eqnarray}
    S = YY^\top/n.
\end{eqnarray}

Let $\lambda^2$ denote the leading eigenvalue of  $S$, and let
\begin{equation}
    \ell^2 = \frac{tr(S) - \lambda^2}{n-1}
\end{equation}
be the average of the non-zero eigenvalues that are less than $\lambda^2$, where $tr(S)$ denotes trace. We can think of 
\[
\eta^2 = \lambda^2 - \ell^2
\]
as the average leading sample eigengap, and it turns out that $\eta^2$ is an unbiased approximation of $\zeta^2$.

To approximate the vector $b = \beta/||\beta||$, we could select the
leading sample eigenvector $h$ of $S$.
However, for fixed $n$, the asymptotic limit in $p$ of the cosine of the angle beetween $h$ and $b$ is positive. In fact it is equal to
\[
\big(1 + \frac{\delta^2}{n B^2 \sigma^2}\big)^{-1} <1,
\]
where $B^2$ is the limit of $||\beta||^2/p$ and $\delta^2$ is the limit of $(1/p) \sum_{i=1}^p \delta_i^2$ as $p \to \infty$.

A definite improvement is obtained with the James-Stein eigenvector shrinkage estimator, $\hjse$, defined as follows.  Let 
$h_1$ denote the projection of $h$ onto the line spanned by $\bfo$. Define the shrinkage constant
\begin{equation}
    c = \frac{\ell^2}{\lambda^2 (1-||h_1||^2)}
\end{equation}
and
\begin{equation}
    H = c h_1 + (1-c) h.
\end{equation}
Then
\begin{equation}
    \hjse = H/||H||.
\end{equation}
The unit vector $\hjse$ is obtained from $h$ by correcting a concentration of measure effect that pushes $h$ farther away from $\bfo$ than $b$.

Recalling that the $i$th diagonal element $s_{ii}$ of $S$ is the sample variance of the $i$th asset, we estimate the $i$th idiosyncratic variance as
\begin{equation}
    \hat \delta_i^2 = s_{ii} - \eta^2 (\hjse_i)^2.
\end{equation}
Our estimated non-singular covariance matrix $\sjse$ is now
\begin{equation}
    \sjse = \eta^2 \hjse(\hjse)^\top + diag(\hat \delta_1^2,\dots,\hat \delta_p^2).
\end{equation}

\subsection{Consistency of single index estimators} \label{subsec:consistency}

When estimating betas from market data as in Section \ref{subsec:beta_est}, there are three covariance matrices in the picture:
\begin{enumerate}
    \item the unobserved population covariance $$\Sigma = \sigma^2 \beta \beta^\top + \Delta,$$
    \item the covariance estimated from observed returns
    $$ \sjse = \eta^2 \hjse(\hjse)^\top + \djse,$$
    \item the market covariance incorporating the observed market portfolio $w_M$
    $$\Sigma^M = \sigma^2_M \beta^M (\beta^M)^\top + \Delta^M$$
    with notation defined in \eqref{eq:mkt_sigma},  \eqref{eq:mkt_beta}, and \eqref{eq:delta-m}.
\end{enumerate}
There is a consistency question with this approach, because 
$w_M^\top \Sigma^M w_M \neq \sigma_M^2$.  Further, the beta factors for the two estimators disagree: $(\eta/\sigma_M)\hjse \neq \beta^M$.

However, a single factor estimator like $\sjse$ has the benefit that the inconsistency above vanishes for a large number of assets, as $p \to \infty$.

Direct computation establishes the following

\begin{theorem}
Consider a sequence in $p$ of $p$-dimensional population covariance models 
$$\Sigma = \sigma^2 \beta \beta^\top + \Delta$$
where $\sigma^2$ is fixed, $\Delta$ is bounded in $p$, and $|\beta|^2/p$ tends to a positive finite limit. Suppose the number of observations used to determine the sample covariance matrix is fixed independent of $p$.

Then, with $\Sigma^M$ computed using $\hat \Sigma = \sjse$, we have, as $p \to \infty$,
\begin{eqnarray}
    |b^M - \hjse| \to 0 \\
    | \sigma^2_M |\beta^M|^2 - \eta^2 |/p \to 0 \\
    |\delta^M_i - \djse_i| \to 0.
\end{eqnarray}
\end{theorem}

This means that the market beta, if computed with $\sjse$, is asymptotically the same as the leading factor in the single index covariance matrix that defines it.

\section{Multifactor JSM estimation} \label{subsec:JSM}

The multifactor JSM estimator is obtained from a PCA estimate by an appropriate shrinkage of the $q$-dimensional factor subspace toward a suitable shrinkage target. It is suitable when $p >> n$; complete discussion of this method appears in \cite{skggb2024}.

In this section we summarize the method, which is the same for any $q$. The returns model is
\begin{equation}
    r = Bf + \epsilon,
\end{equation}
where $f$ is a mean-zero random $q$-vector of returns to $q$ risk factors, $\epsilon$ is a mean zero random $p$-vector of specific returns whose components are independent of each other and of $f$, and $B$ is an unknown $p \times q$ parameter matrix of sensitivities of the securities to the factors.

With a time series of $n$ i.i.d.~returns, $q < n < p$, let $R$ denote the $p \times n$ matrix of centered returns data. For the sample covariance matrix $S = RR^\top/n$ with rank $n_+ \leq n$ and positive eigenvalues $\lambda_1^2 \geq \lambda_2^2 \geq  \lambda_3^2 \geq \dots \geq \lambda_{n_+}^2$, take the spectral decomposition
\begin{equation}
    S = \sum_{i=1}^{n_+} \lambda_i^2 h_i h_i^\top = HH^\top + N,
\end{equation}
where $h_i$ is the unit eigenvector corresponding to $\lambda_i^2$, $H$ is the $p \times q$ matrix with columns $\lambda_1 h_1,\dots,\lambda_q h_q$, and $N = S - HH^\top$.
Then form the specific variance estimate $\Delta = diag(N)$ to be the diagonal matrix with the same diagonal as $N$.

The $q$-factor PCA covariance estimator is $\Sigma^{\rm PCA} = HH^\top + \Delta$, but this can be improved by shrinking H to obtain
\begin{equation}
    \sjsm = \hjsm  \hjsm^\top + \Delta
\end{equation}
as follows.

From the re-weighted data matrix  $R_w = \Delta^{-1/2}R$, we can recompute the $p \times q$ leading eigenvector matrix $H$ via
\begin{equation}
    R_w R_w^\top/n = H_w H_w^\top + N_w
\end{equation}
and then let $\bar H = \Delta^{1/2} H_w$.
 Define a $p \times q$ shrinkage target
 \begin{equation}
     M = \bfo_p(\bfo_p^\top \Delta^{-1} \bfo_p)^{-1} \bfo_p^\top \Delta^{-1} \bar H.
 \end{equation}
 Then $\hjsm$ is defined by linear shrinkage toward $M$:
\begin{equation}
    \hjsm = \bar H C + M(I - C)
\end{equation}
where 
\begin{equation}
    C = I - \nu^2 J^{-1}, \quad J = (\bar H - M)^\top \Delta^{-1} (\bar H - M),
\end{equation}
and 
\begin{equation}
    \nu^2 = \frac{trace(\Delta)}{n_+-q}.
\end{equation}

The format
\begin{equation}
    \hjsm = B \Omega B^\top + \Delta
\end{equation}
is obtained by setting the columns of $B$ to be the ordered unit eigenvectors of $\hjsm {\hjsm}^\top$, and $\Omega$ the diagonal matrix of corresponding eigenvalues.

\bibliography{ref2} 

\end{document}